\documentclass[pra,twocolumn, showpacs,preprintnumbers,amsmath,amssymb,nofootinbib]{revtex4}

\usepackage{pstricks}
\usepackage{graphicx, color, graphpap}
\usepackage{amsmath}
\usepackage{amssymb}
\usepackage{amsthm}
\usepackage{multirow}
\usepackage{hyperref}

\long\def\ca#1\cb{} 

\ca
 \def\outl#1{}  \def\xa{} \def\xb{}  
\cb

\ca
 \def\outl#1{\par{\medskip\noindent\hspace*{.5cm}\bf
      \mathversion{bold}#1\mathversion{normal}\smallskip} }
 \long\def\xa#1\xb{}
 
\cb

 \def\outl#1{\par{\medskip\noindent\hspace*{.5cm}\bf
      \mathversion{bold}#1\mathversion{normal}\smallskip} }
 \def\xa{} \def\xb{}  

\newcommand{\ket}[1]{|#1\rangle}               
\newcommand{\bra}[1]{\langle #1|}              
\newcommand{\dya}[1]{\ket{#1}\bra{#1}}
\newcommand{\dyad}[2]{\ket{#1}\bra{#2}}        
\newcommand{\ip}[2]{\langle #1|#2\rangle}      
\newcommand{\HS}{\text{HS}}

\newcommand{\BC}{\mathcal{B}}
\newcommand{\CC}{\mathcal{C}}
\newcommand{\DC}{\mathcal{D}}
\newcommand{\EC}{\mathcal{E}}
\newcommand{\FC}{\mathcal{F}}

\newcommand{\HC}{\mathcal{H}}
\newcommand{\IC}{\mathcal{I}}

\newcommand{\NC}{\mathcal{N}}
\newcommand{\PC}{\mathcal{P}}
\newcommand{\QC}{\mathcal{Q}}

\newcommand{\SC}{\mathcal{S}}
\newcommand{\TC}{\mathcal{T}}
\newcommand{\UC}{\mathcal{U}}

\newcommand{\DB}{\mathbb{D}}
\newcommand{\Tr}{{\rm Tr}}

\renewcommand{\geq}{\geqslant}
\renewcommand{\leq}{\leqslant}
\newcommand{\mte}[2]{\langle#1|#2|#1\rangle }
\newcommand{\mted}[3]{\langle#1|#2|#3\rangle }

\newcommand{\ot}{\otimes}
\newcommand{\ad}{^\dagger}

\newcommand{\dl}{\delta }
\newcommand{\Dl}{\Delta }

\newcommand{\lm}{\lambda }

\newcommand{\sg}{\sigma }

\newcommand{\Om}{\Omega }

\newtheoremstyle{example}{\topsep}{\topsep}%
{}
{}
{\bfseries}
{.}
{   }
{\thmname{#1}\thmnumber{ #2}}
\theoremstyle{example}

\newtheorem{theorem}{Theorem}

\newtheorem{lemma}[theorem]{Lemma}
\newtheorem{corollary}[theorem]{Corollary}

\begin{document}

\title{Unification of different views of decoherence and discord}
\author{Patrick J. Coles}
\affiliation{Department of Physics, Carnegie Mellon University, Pittsburgh,
Pennsylvania 15213, USA}


\begin{abstract}
Macroscopic behavior such as the lack of interference patterns has been attributed to ``decoherence", a word with several possible definitions such as (1) the loss of off-diagonal density matrix elements, (2) the flow of information to the environment, (3) the loss of complementary information, and (4) the loss of the ability to create entanglement in a measurement. In this article, we attempt to unify these distinct definitions by providing general quantitative connections between them, valid for all finite-dimensional quantum systems or quantum processes. The most important application of our results is to the understanding of quantum discord, a measure of the non-classicality of the correlations between quantum systems. We show that some popular measures of discord measure the information missing from the purifying system and hence quantify \emph{security}, which can be stated operationally in terms of distillable secure bits. The results also give some strategies for constructing discord measures.
\end{abstract}

\maketitle

\section{Introduction}\label{sct1}

A modern challenge in quantum physics is to explain the macroscopic phenomena seen by everyday observation using the quantum laws that appear to be correct on the small scale. The theory of decoherence \cite{Joos03, ZurekReview, Schloss07} has made major progress in this direction. Yet, interestingly, decoherence can be described in different ways, consider the following possible definitions for decoherence:\\

(D1) Loss of off-diagonal elements of the system's reduced density matrix\\

(D2) Flow of information to the environment\\

(D3) Loss of complementary information from the system (i.e., loss of interference)\\

(D4) Loss of the ability to create entanglement in a projective measurement\\

The first definition is well-known \cite{Joos03, ZurekReview, Schloss07, HornDecIntro08} and, e.g., has been linked to the loss of interference in a two-slit interferometer (TSI) (e.g.\ \cite{Schloss07}). The second definition \cite{Joos03, ZurekReview, Schloss07, HornDecIntro08, OllPouZurPRL10, RBKZurek06, ZwoQuaZur09, RiedelZurekPRL10}, also linked to loss of interference in a TSI \cite{Feynman70, WooZur1979}, has been connected to objectivity when many copies of information go off to the environment \cite{OllPouZurPRL10}, though the present article is only concerned with whether or not a single copy exists in the environment. In (D3), our notion of complementary information will become more precise later, we simply note that the complementary information is the kind of information that is directly responsible for interference (which-phase information). The last definition (D4) is motivated by recent studies of quantum correlations \cite{PianiEtAl11, StrKamBru11, GharEtAl2011, PianiAdesso2011}, and though it may be the least familiar, it can be made intuitive. We think of a projective measurement as a ``test process" meant to probe the system's degree of decoherence; a thought experiment that asks how much entanglement \emph{would} be created if hypothetically we did such a measurement. For example, suppose the system is a qubit in the state $\ket{+}=(\ket{0}+\ket{1})/\sqrt{2}$ and we performed a measurement in the $\{\ket{0}, \ket{1}\}$ basis (modeled as a CNOT with the register initially in the $\ket{0}$ state). Then the state evolves according to $\ket{+}\ket{0}\to (\ket{00}+\ket{11})/\sqrt{2}$, creating a full Einstein-Podolsky-Rosen (EPR) pair of entanglement, since the $\ket{+}$ state was not decohered at all. Instead suppose the system was fully decohered, in a maximally mixed state, then the system and register evolve into a separable state $(\dya{00}+\dya{11})/2$ (i.e., no entanglement).

Looking at the list (D1)--(D4), one cannot help but ask, do they all represent the same thing? A very large amount of intuition suggesting especially that (D1)--(D3) are \emph{quantitatively} connected has been obtained through models of partial decoherence \cite{Joos03, ZurekReview, Schloss07, HornDecIntro08}, which involve an assumption about the form of the coupling between the system and environment (e.g.\ see p. 48 of \cite{Joos03}). Some general connections between these definitions have also been noted, e.g.\ between (D1) and (D2) on p. 47 of \cite{Schloss07}, (D1) and (D4) in \cite{StrKamBru11, PianiEtAl11}. It would be nice to have the general connections, if they exist, systematically worked out, so that we can be confident that we all mean the same thing when we say ``decoherence". The main goal of this article is to present quantitative connections between the four definitions of decoherence given above under very general circumstances, i.e., without invoking a Hamiltonian model. When different phenomena can be shown to scale quantitively with each other for all quantum processes, it suggests that they are simply different \emph{views} of one single phenomenon. 

Let us use the TSI \cite{Schloss07, Feynman70, WooZur1979} to illustrate the connection. Identify $\ket{0}$ and $\ket{1}$ with the particle going through the upper and lower slits respectively. Suppose the which-path information, $\ket{0}$ or $\ket{1}$, is obtained by a photon that scatters off the particle just after it passes through the slit, process (D2). Then off-diagonal elements of the particle's density matrix disappear in this basis (D1), the \textit{which-phase} information, $(\ket{0}+e^{i\phi}\ket{1})/\sqrt{2}$ or $(\ket{0}-e^{i\phi}\ket{1})/\sqrt{2}$, responsible for the interference pattern is lost (D3), and if a \textit{second} photon were to scatter off the particle, it would not become entangled with this degree of freedom of the particle (D4).

In this example we focused on off-diagonal elements in the which-path basis, but we could have discussed a different basis. The power of our main results lies in the fact that they apply to every orthonormal basis of the Hilbert space, not just the ``most classical" (or ``pointer") basis. When possible, we state our results even more generally for any set of orthogonal projectors $Z=\{Z_j \}$ that decompose the system's identity operator $I=\sum_j Z_j$ \cite{Gri07}. The special case where the $Z_j$ are rank-one corresponds to an orthonormal basis, but the generalization to coarse-grained projectors is crucial for macroscopic systems \cite{CQT}. Our strongest results relate (D1), (D2), and (D4) for arbitrary $Z$, with information-theoretic equations that connect the distance of the state to one with no $Z$ off-diagonal elements to the $Z$ information missing from the environment, and in turn, to the entanglement created in a $Z$ measurement. We also give a somewhat weaker connection to (D3), valid when $Z$ is a basis.

The most important application of the connections we find is to the understanding of \textit{quantum discord}, a measure of non-classical correlations that was originally introduced to study decoherence and the emergence of classicality \cite{OllZur01, ZurekReview}. Discord has since caught the attention of the quantum information community, with the intriguing idea that it may be more appropriate than entanglement to quantify how useful a resource is for quantum computing \cite{DatShaCav08}. Various interpretations and alternative measures for discord have been found; we especially recommend a recent review article on this topic \cite{ModiEtAl2011review}. Our unification of the decoherence views allows us to give a new interpretation, based on (D2), for several popular measures of discord: the one-way information deficit \cite{ZurekDemons03, HorEtAl05}, geometric discord \cite{DakVedBruPRL10}, geometric entanglement discord \cite{StrKamBru11}, and relative entropy of quantumness \cite{HorEtAl05} can all be interpreted as measuring the \textit{minimum information missing from the purifying system}. If one imagines an adversary holding possession of the purifying system, then the task of distilling secure classical bits (secure from the adversary) is in fact operationally characterized by some of these popular measures of discord.

The rest of the article is organized as follows. Section~\ref{sct2} introduces our notation, including our quantitative measures of decoherence. Section~\ref{sct3} presents our main results. Section~\ref{sct4} discusses our results within the context of two decoherence paradigms: (1) the system-environment at a single time slice and (2) at two different time slices. Section~\ref{sct5} discusses how our results imply that certain information-processing tasks are connected in the asymptotic limit. Section~\ref{sct6} takes up quantum discord, discussing some popular discord measures that were constructed based on (D1) or (D4), and then giving our contribution in terms of discord measures based on (D2) or (D3).

\section{Notation}\label{sct2}

\subsection{Classes of quantum states}\label{sbct2.1}

Since we will be discussing entanglement and discord, it will useful to define certain classes of states. On some Hilbert space $\HC_A$, we will denote the set of all normalized density operators as $\NC_A$, sometimes just writing $\NC$, dropping the subscript when the space is obvious. For a bipartite space $\HC_{AB}$, we will use $\SC$ to denote the set of all separable states, of the general form: 
\begin{equation}
\label{eqn1}
\rho_{AB}=\sum_j p_j \rho_{A,j}\ot \rho_{B,j}
\end{equation}
with $\rho_{A,j}\in \NC_A$, $\rho_{B,j}\in \NC_B$, and $\{p_j\}$ some probability distribution. Likewise we define classical states on $\HC_{AB}$, or more precisely classical-classical states, denoted $\CC\CC$, as those with the general form:
\begin{equation}
\label{eqn2}
\rho_{AB}=\sum_{j,k} p_{j,k}\dya{j}\ot \dya{k}
\end{equation}
where $\{\ket{j}\}$ and $\{\ket{k}\}$ are orthonormal bases on $\HC_A$ and $\HC_B$ respectively, and $\{p_{j,k}\}$ is a (joint) probability distribution. Finally, we will use $\CC\!\QC$ to denote the set of classical-quantum states, of the form
\begin{equation}
\label{eqn3}
\rho_{AB}=\sum_{j} p_{j}\dya{j}\ot \rho_{B,j}.
\end{equation}
It is clear that the following relationship holds for the state classes defined above: $\CC\CC\subset \CC\!\QC \subset \SC\subset \NC$.

\subsection{Quantitative measures}\label{sbct2.2}

In this section, we introduce the measures that we will use to quantify decoherence. We first note the following useful concept, which has been employed previously to study correlations in quantum systems \cite{Gri05, Gri07, GLGr10, ColesEtAl}.

\textit{Definition:} A type of information $Z=\{Z_j\}$ about system $A$ is a decomposition of the identity into a set of orthogonal projectors, $I_A = \sum_j Z_j$. (We shall say $Z$ is an orthonormal basis if the $Z_j$ are rank-one.)

For a tripartite state $\rho_{ABC}$, we will want to quantify how much of the $Z$ type (or other types) of information about $A$ is ``located inside" some other system ($B$ or $C$) in the sense that some observable of this other system provides information about $Z$.  To do so, it is helpful to think of a $Z$ measurement on system $A$, modeled as an isometry:
$$V_Z=\sum_j \ket{j}_{M_Z} \ot Z_j,$$
which stores the measurement outcomes in the (orthonormal) basis states $\{\ket{j}\}$ of a register system $M_Z$, where the $Z_j$ act on system $A$ and it is implicit that identity acts on systems $B$ and $C$. (We think of $M_Z$ as like a $Z$-measurement device.) The post-measurement state is $\tilde{\rho}_{M_Z ABC}:=V_Z\rho_{ABC}V_Z\ad$, and
\begin{equation}
\label{eqn4}
\tilde{\rho}_{M_Z C}=\Tr_{AB} (\tilde{\rho}_{M_Z ABC})=\sum_j p_j \dya{j}\ot\rho_{C,j}
\end{equation}
is a $\CC\!\QC$ state. The nature of this $\CC\!\QC$ state tells us how much $Z$ information is located in (or known to) system $C$, for example, consider the following two extreme cases.

\textit{Definition:} A type of information $Z$ about $A$ is said to be perfectly present in $C$ if all the $\rho_{C,j}$ in \eqref{eqn4} are orthogonal (i.e., an appropriate measurement on $C$ will perfectly extract the $Z$ information).

\textit{Definition:} A type of information $Z$ about $A$ is said to be completely absent (or perfectly secure) from $C$ if it is uniformly distributed, $p_j=1/|Z|, \forall j$, and all the $\rho_{C,j}$ in \eqref{eqn4} are identical (and hence equal to $\rho_C$).

We would like to move off of these extreme cases and say something quantitative about how much $Z$ information is in $C$; for this purpose, the conditional entropy of the state $\tilde{\rho}_{M_Z C}$ provides a reasonable measure. In this article, we focus on three different conditional entropies: (1) the von Neumann entropy, since it is the most familiar, (2) a quadratic approximation of the von Neumann entropy (defined below), since it is easy to calculate and manipulate, and (3) the min-entropy \cite{RennerThesis05}, since it has a nice operational meaning \cite{KonRenSch09}. The von Neumann conditional entropy \cite{NieChu00} of $\tilde{\rho}_{M_Z C}$ is denoted:
\begin{equation}
\label{eqn5}
H(Z|C):=H(\tilde{\rho}_{M_Z C})-H(\rho_C),
\end{equation}
where $H(\rho)=-\Tr(\rho\log_2 \rho)$. We specifically use the notation $H(Z|C)$ because this quantity measures the information about $Z$ that is missing from $C$, or the $Z$-uncertainty given $C$. In other words, we think of $H(Z|C)$ as a \emph{classical} entropy conditioned on quantum side information, and this interpretation is operationally justified \cite{DevWin03}. This quantity is bounded by $0\leq H(Z|C)\leq \log_2 N$, where $N=|Z|$ is the number of elements in $\{Z_j\}$. Furthermore, $H(Z|C)= 0$ iff the $Z$ information is perfectly present in $C$, and $H(Z|C)= \log_2 N$ iff the $Z$ information is completely absent from $C$.

A quadratic approximation of the von Neumann conditional entropy is:
\begin{equation}
\label{eqn6}
H_Q(Z|C):=\Tr(\rho_{C}^2)-\Tr(\tilde{\rho}_{M_Z C}^2),
\end{equation}
which is obtained by replacing $H(\rho)$ with $H_Q(\rho)=1-\Tr(\rho^2)$, a quadratic function of $\rho$ that is often called linear entropy. A discussion of the properties of this quadratic conditional entropy can be found in Appendix~\ref{app1}.

As noted in \cite{KonRenSch09}, the conditional min-entropy of a $\CC\!\QC$ state, in our case $\tilde{\rho}_{M_Z C}$, can be written as
\begin{align}
\label{eqn7}
H_{\min}(Z|C)&=-\log_2 p_{\text{guess}}(Z|C),\notag\\
p_{\text{guess}}(Z|C)&= \max_{Q} \sum_j p_j \Tr(Q_j\rho_{C,j}),
\end{align}
where $p_{\text{guess}}(Z|C)$ is the probability for ``Charlie" to guess $Z$ correctly with an optimal POVM (Positive Operator Valued Measure) $Q=\{ Q_j\}$ on $C$. The min-entropy shares many of the same properties as the von Neumann entropy, with $0\leq H_{\min}(Z|C)\leq \log_2 N $, and $H_{\min}(Z|C)$ achieving these upper and lower bounds under the same conditions as $H(Z|C)$ does, as discussed in Appendix~\ref{app2}.

We find it useful to introduce the following measures of \emph{certainty}, measuring the opposite of entropy:
\begin{align}
\label{eqn8}
\CC(Z|C)&:= \log_2 N - H(Z|C),\notag\\
\CC_{Q}(Z|C)&:= (N-1)\Tr(\rho_C^2)-N H_{Q}(Z|C), \notag\\
\CC_{\min}(Z|C)&:= \log_2 N- H_{\min}(Z|C),
\end{align}
Each of these certainty measures is non-negative, and vanishes iff the $Z$ information is completely absent from $C$.

As we will see below (Theorem~\ref{thm1}), the three conditional entropies defined above are connected, respectively, to the following measures of distance or distinguishability between two density operators: (1) the relative entropy, (2) the Hilbert-Schmidt distance, and (3) the fidelity, which are respectively given by:
\begin{align}
D(\rho ||\sg)&=-H(\rho)-\Tr(\rho\log_2 \sg)\notag\\
D_{\HS}(\rho ,\sg)&=\Tr[(\rho-\sg)^2]\notag\\
F(\rho , \sg)&=[\Tr (\sqrt{\rho}\sg\sqrt{\rho})^{1/2}]^2.\notag
\end{align}

We will also see below (Theorem~\ref{thm2}) that some of the conditional entropies above are connected to measures of \textit{entanglement}. The von Neumann entropy is connected to both the distillable entanglement $E_D$ (see \cite{BennetEtAl1996,DevWin05,HHHH09} for the definition) and the relative entropy of entanglement $E_R$ \cite{VedrPlen1998}; whereas, the min-entropy is connected to the geometric entanglement $E_G$ \cite{WeiGol2003,StreltsovEtAl2010NJP}. The latter two measures can be written, for some bipartite state $\rho_{AB}$, as:
\begin{align}
E^{A|B}_R(\rho_{AB})&=\min_{\sg_{AB}\in\SC}D(\rho_{AB} || \sg_{AB}),\notag\\
E^{A|B}_G(\rho_{AB})&=\min_{\sg_{AB}\in\SC}[1-F(\rho_{AB} , \sg_{AB})].\notag
\end{align}

\section{Main Results}\label{sct3}

In this section, we give some quantitative relations between the four definitions of decoherence. These relations then imply the equivalence of the classicality (i.e., complete decoherence) conditions associated with the four definitions. Our first theorem connects (D1), quantified by the distance of the state to a state with no off-diagonal elements, to (D2), quantified by a conditional entropy. The proof is in Appendix~\ref{app3}. We note that a version of Eq.~\eqref{eqn9} was first presented in \cite{ColesEtAl} without stating its connection to decoherence.
\begin{theorem}
\label{thm1}Let $Z=\{Z_j\}$ be general type of information about $A$, and let $\rho_{ABC}$ be pure, then
\begin{align}
\label{eqn9}&\text{(i)   }H(Z|C)=\min_{\sg_{AB}\in\NC} D(\rho_{AB}|| \sum_j Z_j \sg_{AB} Z_j) \\
\label{eqn10}&\text{(ii)   }H_Q (Z|C)= \min_{\sg_{AB}\in\NC} D_{\HS}(\rho_{AB}, \sum_j Z_j \sg_{AB} Z_j)\\
\label{eqn11}&\text{(iii)   }p_{\text{guess}}(Z|C)= \max_{\sg_{AB}\in\NC} F(\rho_{AB}, \sum_j Z_j \sg_{AB}Z_j).
\end{align}\end{theorem} \openbox

The right-hand-sides of \eqref{eqn9} and \eqref{eqn10} can be replaced, respectively, with the following:
\begin{align}
\label{eqn12}
D(\rho_{AB}&|| \sum_j Z_j \rho_{AB}Z_j),\notag\\
D_{\HS}(\rho_{AB}&, \sum_j Z_j \rho_{AB}Z_j),
\end{align}
since $\sum_j Z_j \rho_{AB}Z_j $ is the state that accomplishes the minimization \cite{ModiEtAl2010, LuoFu2010, ModiEtAl2011review}. As discussed in Appendix~\ref{app4}, Eq.~\eqref{eqn10} can also be written:
\begin{equation}
\label{eqn13}
H_Q(Z|C)= \sum_{j,k\neq j}\| Z_j\rho_{AB}Z_k \|^2,
\end{equation}
where $\| M\|^2=\Tr(M\ad M)$ is the Hilbert-Schmidt norm. This gives a more direct connection of off-diagonals elements to information in the environment.

We note that a special case of Thm.~\ref{thm1} is the case of pure $\rho_{AB}$, for which one takes $C$ to be a trivial system, and the left-hand-sides of \eqref{eqn9}, \eqref{eqn10}, and \eqref{eqn11} respectively reduce to $H(Z)=-\sum_j p_j \log_2 p_j$, $H_Q(Z)=1-\sum_j p_j^2$, and $p_{\text{guess}}(Z)=\max_j p_j $, with $p_j=\Tr(Z_j \rho_A)$. In this case, these are just the corresponding classical entropies (or certainty for $p_{\text{guess}}$) of the $Z$ random variable.

The next theorem connects (D4) to (D2), and in turn to (D1) by the previous theorem. For the post-measurement state $\tilde{\rho}_{M_ZAB}=V_Z\rho_{AB}V_Z\ad$, denote the entanglement between the register $M_Z$ and the $AB$ system as $E_K^{M_Z|AB}(\tilde{\rho}_{M_ZAB})$, which for simplicity we may write as $E_K^{M_Z|AB}$. Here, an appropriate subscript $K$ will refer to the particular entanglement measure. Inspired by \cite{PianiEtAl11} and particularly by \cite{StrKamBru11} (which also considered $E_G$), we find the following result, proved in Appendix~\ref{app5}.
\begin{theorem}
\label{thm2}
Let $Z=\{Z_j\}$ be general type of information about $A$, and let $\rho_{ABC}$ be pure, then
\begin{align}
\label{eqn14}  &\text{(i)   } E^{M_Z|AB}_D=E^{M_Z|AB}_R  =H(Z|C)\\
\label{eqn15}  &\text{(ii)   } E^{M_Z|AB}_G =1-p_{\text{guess}}(Z|C).
\end{align}
\end{theorem} \openbox

Next we wish to connect (D3) to the other definitions of decoherence. To discuss complementary information, we restrict $Z=\{\dya{Z_j}\}$ to be an orthonormal basis, and consider an orthonormal basis $W=\{\dya{W_k}\}$ that is mutually unbiased (MU) w.r.t.\ $Z$ in the sense that $|\ip{Z_j}{W_k}|=1/\sqrt{d_A}, \forall j,k$, where $d_A=\dim (\HC_A)$. We note that, to some degree, entropic uncertainty relations that allow for quantum side information connect (D3) to (D2). For example, the following uncertainty relation holds for the von Neumann entropy \cite{RenesBoileau, BertaEtAl},
\begin{align}
\label{eqn16}
H (Z|C)\geq \CC (W|B),
\end{align}
which reads: the uncertainty about $Z$ given $C$ is lower-bounded by the certainty about $W$ given $B$. A similar sort of uncertainty relation has been obtained for the min-entropy \cite{TomRen2010}, and for other entropies \cite{ColesColbYuZwo2011}. However, such uncertainty relations are \textit{inequalities}, whereas the main goal of this article is to find \textit{equalities}. Only the latter can prove the equivalence of (D3) to the other definitions. To see the issue, note that if $C$ perfectly contains the $Z$ information, then $H (Z|C)=0$, and \eqref{eqn16} implies that $\CC (W|B)= 0$ and hence that the $W$ information is completely absent from $B$. However \eqref{eqn16} does not imply the converse; it does \textit{not} say that if the $W$ information is completely absent from $B$, then $C$ perfectly contains the $Z$ information.

In the following theorem, we formulate an entropic uncertainty \textit{equation}, which allows us to establish a sort-of converse like the one described above. We need, though, to consider an average over an equivalence class of bases that are MU w.r.t.\ $Z$.\footnote{This is related to equivalence classes of complex Hadamard matrices \cite{TadZyc2005}.} An equivalence class refers to all bases that can be made equivalent to each other by the action of a unitary that is diagonal in the $Z$ basis (e.g.\ for qubits, taking $Z$ to be the standard basis, then all orthonormal bases in the $xy$-plane of the Bloch sphere form an equivalence class). The proof is in Appendix~\ref{app6}.
\begin{theorem}
\label{thm3}
Let $\rho_{ABC}$ be pure, let $Z$ be an orthonormal basis on $\HC_A$, let $\BC_Z$ be an equivalence class of orthonormal bases that are MU w.r.t.\ $Z$, then
\begin{equation}
\label{eqn17}H_Q (Z|C)= \langle \CC_Q (W|B)\rangle_{\BC_Z }
\end{equation}
where $\langle \cdot \rangle_{\BC_Z }$ is the average over all bases $W$ in $\BC_Z$.
\end{theorem} \openbox

We note that \eqref{eqn17} holds for all equivalence classes $\BC_Z$, so the right-hand-side must be the same for all equivalence classes. One can therefore replace the average in \eqref{eqn17} with the average over all equivalence classes, which is essentially an average over \emph{all} orthonormal bases that are MU w.r.t.\ $Z$.

With the above quantitative connections, we are able to establish the equivalence of four different classicality conditions.
\begin{corollary}
\label{thm4}
The following classicality conditions are equivalent, for any tripartite pure state $\rho_{ABC}$ and any orthonormal basis $Z$ on $\HC_A$. (The equivalence of conditions (i)--(iii) holds more generally for any type of information $Z$ about $A$.) \\
(i) $\rho_{AB}=\sum_j Z_j \rho_{AB} Z_j$.\\
(ii) The information about $Z$ is perfectly present in $C$.\\
(iii) A $Z$-measurement produces no entanglement between the measuring device and the $AB$ system.\\
(iv) The information about an equivalence class of bases that are MU w.r.t.\ $Z$ is completely absent from $B$.
\end{corollary}
\begin{proof}
Since $D(\rho ||\sg)=0 \Leftrightarrow \rho=\sg$, Theorem~\ref{thm1} implies that (i) $\Leftrightarrow H(Z|C)=0 \Leftrightarrow $ (ii). Theorem~\ref{thm2} then implies that (ii) $\Leftrightarrow H(Z|C)=0 \Leftrightarrow E^{M_Z|AB}_R=0 \Leftrightarrow $ (iii). (Note that the argument, up until now, did not depend on $Z$ being an orthonormal basis.) Finally, Theorem~\ref{thm3} implies that (ii) $\Leftrightarrow H_Q (Z|C)=0 \Leftrightarrow \langle \CC_Q (W|B)\rangle_{\BC_Z } = 0 \Leftrightarrow $ (iv). For this last step, we used the fact that $\CC_Q (W|B)\geq 0$, with equality iff the $W$ information is completely absent from $B$, see Appendix~\ref{app1}.
\end{proof}

\section{Decoherence paradigms}\label{sct4}

\subsection{Introduction}\label{sbst4.0}

Here we apply our main results to two decoherence paradigms. When reading the following, it is helpful to keep in mind a physical situation of interest. For example, one can imagine the TSI discussed in the Introduction, though for simplicity we will consider below the Mach-Zehnder interferometer \cite{GerryKnight05}, an interferometer for single photons. As depicted in Fig.~\ref{fgr1}, after the first beam splitter, the photon can either go through the upper or lower arm, respectively identified as the states $\ket{0}$ and $\ket{1}$. If it goes through the lower arm, it receives a phase shift $e^{i\phi}$ before impinging on a second beam splitter. The photon is detected in one of two possible detectors placed after the second beam splitter. Varying $\phi$ causes the probability to detect the photon in the upper detector to vary sinusoidally, i.e., producing an interference pattern. This interference pattern can be altered by the presence of an environment $E$ within the interferometer (e.g., $E$ could be a gas whose fluctuations randomly alter the local refractive index \cite{Pater2005}, imparting random phase-shifts and reducing the fringe visibility.) We will make use of this interferometer later to illustrate how our results apply to a common decoherence paradigm. 

\subsection{System-environment at a single time}\label{sbst4.1}

Consider the important decoherence paradigm of a bipartite cut of the Universe at a single point in time. The two parts are the system of interest $S$ and the environment $E$, which may have already interacted, and so we assume they are described by a general bipartite pure state $\rho_{SE}$. To apply the results in Sect.~\ref{sct3}, we set $A=S$, $C=E$, and $B$ to a trivial (one-dimensional) system. Then, Thms.~\ref{thm1} and~\ref{thm2} say, for example, that:
\begin{align}
\label{eqn18}
p_{\text{guess}}(Z|E)&= \max_{\sg_{S}\in\NC} F(\rho_{S}, \sum_j Z_j \sg_{S}Z_j)\notag\\
&= 1-E^{M_Z|S}_G(\tilde{\rho}_{M_ZS}).
\end{align}
The extreme case of complete decoherence, where the system's reduced density operator $\rho_{S}$ has no $Z$ off-diagonal terms, corresponds to $F(\rho_{S}, \sum_j Z_j \rho_{S}Z_j)=1$. From \eqref{eqn18}, this implies that $p_{\text{guess}}(Z|E)=1$, i.e.\ given access to the environment, one can in principle perfectly guess the ``$Z$-component" of the system. This also implies that $E^{M_Z|S}_G(\tilde{\rho}_{M_ZS})=0$, in other words, no entanglement will be created if one does a $Z$ measurement on the system. 

However, the power of \eqref{eqn18} lies in the fact that it holds for the case of \emph{partial} decoherence, and for \emph{every} type of information $Z$ about the system. Such a quantitative connection truly \emph{unifies} the different definitions (D1), (D2), and (D4) of decoherence. The quantity $p_{\text{guess}}(Z|E)$ gives a simple operational measure for decoherence, as the probability to guess $Z$ correctly given access to the environment, and in some situations, it might be easier to calculate than the other quantities in \eqref{eqn18}. Let us consider a few examples.

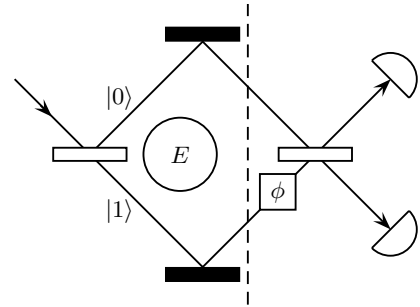
\begin{figure}[t]
\begin{center}
\begin{pspicture}(-2,-1.8)(6,1.8) 
\newpsobject{showgrid}{psgrid}{subgriddiv=1,griddots=10,gridlabels=6pt}
\def\lwd{0.025} 
\def\lwb{0.10}  
\def\mirw{0.5}\def\mirt{0.2}  
\psset{
labelsep=2.0,
arrowsize=0.150 1,linewidth=\lwd}
  \def\ruarr(#1){\rput(#1){\psline{->}(-0.15,-0.15)(0,0)}}  
  \def\rdarr(#1){\rput(#1){\psline{->}(-0.15,+0.15)(0,0)}}  
\def\brecs(#1,#2){%
\psframe[fillcolor=white,fillstyle=solid,linestyle=none](-#1,-#2)(#1,#2)}
\def\circb{
\pscircle[fillcolor=white,fillstyle=solid]{0.25}}
\def\crhatch#1{\pspolygon[linestyle=none,fillstyle=hlines,hatchangle=30,%
hatchwidth=.02,hatchsep=0.15]#1}
\def\rdet{0.35}  
\def\detect{
\psarc[fillcolor=white,fillstyle=solid](0,0){\rdet}{-90}{90}
\psline(0,-\rdet)(0,\rdet)}
\def\rdet{0.35} \def\sdet{0.247}
\def\detectt{
\psarc[fillcolor=white,fillstyle=solid](0,0){\rdet}{-45}{135}
\psline(-\sdet,\sdet)(\sdet,-\sdet)}
\def\detectu{
\psarc[fillcolor=white,fillstyle=solid](0,0){\rdet}{-135}{45}
\psline(-\sdet,-\sdet)(\sdet,\sdet)}
   \def\rdot{0.1} \def\rodot{0.2} 
\def\dot{\pscircle*(0,0){\rdot}} 
\def\odot{\pscircle[fillcolor=white,fillstyle=solid](0,0){\rodot}} 
\def\dput(#1)#2#3{\rput(#1){#2}\rput(#1){#3}}
\def\rectg(#1,#2,#3,#4){
\psframe[fillcolor=black,fillstyle=solid](#1,#2)(#3,#4)}
\def\rectc(#1,#2){%
\psframe[fillcolor=white,fillstyle=solid](-#1,-#2)(#1,#2)}
\def\hdg{0.25} \def\squ{%
\psframe[fillcolor=white,fillstyle=solid](-\hdg,-\hdg)(\hdg,\hdg)}
   \def\tkw{0.1} 
\def\htick{\psline(-\tkw,0)(\tkw,0)}
\def\vtick{\psline(0,-\tkw)(0,\tkw)}
\def\mirra{\rectg(-\mirw,0,\mirw,\mirt)}
\def\mirrb{\rectg(-\mirw,-\mirt,\mirw,0)}
\def\wedge{\pspolygon[fillcolor=white,fillstyle=solid]%
(-.5,0)(0.7,0.7)(0.7,-0.7)}
\psline{->}(-1,1)(1.5,-1.5)(4,1)
\psline{->}(0,0)(1.5,1.5)(4,-1)
\rdarr(-0.5,0.5)
\psline[linestyle=dashed](2.1,-2)(2.1,2)
\rput(0,0){\rectc(.5,.1)}
\rput(3,0){\rectc(.5,.1)}
\rput(1.5,1.5){\mirra}\rput(1.5,-1.5){\mirrb}
\rput(0.4,0.75){$\ket{0}$}
\rput(4,1){\detectt}
\rput(4,-1){\detectu}
\rput(0.4,-0.75){$\ket{1}$}
\pscircle(1.2,0){0.5}
\rput(1.2,0){$E$}
\rput(2.5,-0.5){\squ}
\rput(2.5,-.5){$\phi$}
\end{pspicture}

\caption{%
  Two-path interferometer for single photons, with a phase shifter inserted in the lower arm. Identifying $\{\ket{0},\ket{1}\}$ as the which-path basis, the apparatus to the right of the vertical dashed line measures in the basis $(\ket{0}\pm e^{-i\phi}\ket{1})/\sqrt{2}$. An environment $E$ may obtain some which-path information prior to this measurement.\label{fgr1}}
\end{center}
\end{figure}

\textit{Example:} Suppose, at time $t=0$, the system has yet to interact with its environment and so we describe the system with a pure state $\ket{\psi}$. Let $Z$ be an orthonormal basis for $\HC_S$ that includes the state $\ket{\psi}$ as one of its basis elements. Then $p_{\text{guess}}(Z|E)=1$, consistent with $\dya{\psi}$ having no off-diagonals in the $Z$ basis, and creating no entanglement if $Z$ would be measured. Let $W$ be an orthonormal basis for $\HC_S$ that is \textit{unbiased} w.r.t.\ $\ket{\psi}$, i.e.\ $\mte{\psi}{W_j}=1/d_S, \forall j$. Then we find $p_{\text{guess}}(W|E)=1/d_S$ and hence entanglement \emph{would} be generated from a $W $ measurement: $E^{M_W |S}_G(\tilde{\rho}_{M_W S})=1-1/d_S$.

\textit{Example:} Suppose, at some later time $t>0$, the system has become maximally-entangled with its environment. Then for \emph{all} types of information $Z$ about the system, we have $p_{\text{guess}}(Z|E)=1$, and hence no entanglement would be generated from measuring the system.

Let us now consider the connection to (D3) given by Theorem~\ref{thm3}, which we illustrate with the single-photon interferometer in Fig.~\ref{fgr1}. We think of the system as a qubit with $Z$ being the which-path basis, $\{\ket{0},\ket{1}\}$. The apparatus to the right of the dashed line in Fig.~\ref{fgr1} performs a measurement in the $W $ basis, $\ket{W\pm}=(\ket{0}\pm e^{-i\phi}\ket{1})/\sqrt{2}$ \cite{GerryKnight05}. Suppose that immediately after the first beam splitter, the system is in the state $\ket{+}=(\ket{0}+ \ket{1})/\sqrt{2}$, but interaction with an environment $E$ within the interferometer leaves the system described by a density operator $\rho_S=\Tr_E(\rho_{SE})$ at the time slice corresponding to the dashed line. Then we find:
\begin{equation}
\label{eqn19}
|\mted{0}{\rho_S}{1}|^2=\frac{1}{2}H_Q(Z|E)=\frac{1}{2}\langle \CC_Q (W)\rangle_{\BC_Z },
\end{equation}
where the first equation is from \eqref{eqn13} and the second is from \eqref{eqn17}. Here $\CC_Q(W)=1-2 H_Q (W)$ measures our certainty about the outcome of the $W$ measurement, i.e.\ how well we can predict which detector in Fig.~\ref{fgr1} will click. Taking the average $\langle \cdot \rangle_{\BC_Z }$ corresponds to averaging over all possible choices of $\phi$. So \eqref{eqn19} says that our average certainty (allowing $\phi$ to vary) about which detector will click is quantified by the magnitude squared of the off-diagonal element \emph{in the which-path basis}, which in turn is quantified by how little $E$ knows about the which-path information. A drop in $|\mted{0}{\rho_S}{1}|^2$, due to $E$ picking up some which-path information, implies a loss in our ability to predict which detector clicks (``smearing out" the interference pattern), and hence \eqref{eqn19} gives a simple, intuitive connection between (D1), (D2), and (D3).

\subsection{System-environment at two times}\label{sbst4.2}

Let us show how our results can be applied to a second decoherence paradigm, considering the system and environment at two different points in time. This can be viewed as a tripartite pure state as follows. Let the system $S_0$ and environment $E_0$ at time $t_0$ evolve according to a unitary $U$ to time $t_1$ at which point call them $S_1$ and $E_1$. Typically, one assumes \cite{NieChu00} the environment starts in a fixed pure state $\ket{E_0}$; because this never changes, it can be absorbed into the mapping, which turns the map into an isometry $V=U \ket{E_0}$. Then, see Fig.~\ref{fgr2}, introduce a copy $S'_0$ of the $S_0$ system, let $\ket{\Phi}$ be a maximally entangled state on $S_0S'_0$, feed $S'_0$ into $V$ but let $S_0$ evolve freely in time, so that at time $t_1$ one is left with a tripartite pure state $\ket{\Om}=(I\ot V)\ket{\Phi}\in \HC_{S_0S_1E_1}$. The equivalence between the dynamic and static views, as depicted in Fig.~\ref{fgr2}, is often called Choi-Jamiolkowski isomorphism (e.g., see \cite{ColesEtAl}).

While we must convert to the static view to apply our main results, the dynamic view may be more intuitive. In the dynamic view, one considers complementary quantum channels $\EC(\cdot)=\Tr_{E_1}[V(\cdot)V\ad]$ and $\FC(\cdot)= \Tr_{S_1}[V(\cdot)V\ad] $. Here it is common to discuss robust states whose entropy increases very little from the action of $\EC$ \cite{ZurekReview}. Our view is to speak of robust \textit{information} \cite{RBKetalPRL08, RBKetal2010, BenyThesis09, ColesEtAl}. One imagines $S_0$ sending a type of information $\{Z_j\}$ down the $\EC$ channel and asking whether or not it is preserved, i.e., is the set $\{\EC(Z_j)\}$ distinguishable at the output. Measures such as $H(Z|S_1)$ quantify this and can be used to define an information analog of the predictability sieve \cite{ZurekReview}. One can also ask to what degree is $\EC$ a \textit{decohering channel}; insert the marginal $\rho_{S_0S_1}=\Tr_{E_1}(\dya{\Om})=(\IC \ot \EC)\dya{\Phi}$ into \eqref{eqn9}, and \eqref{eqn9} says that the distance of $\EC$ from a channel that destroys $Z$ off-diagonals is measured by how poorly the complementary channel $\FC$ transmits the $Z$ information, $H(Z|E_1)$. The latter can sometimes be calculated fairly easily, as illustrated by the following example.

\textit{Example:} For the qubit phase-flip channel \cite{NieChu00}, $\EC(\rho)=(1-p)\rho+p\sg_Z\rho\sg_Z$ with $0\leq p\leq 1/2$ and $\sg_Z=\dya{0}-\dya{1}$, one finds its distance to a channel that destroys $Z=\{\dya{0},\dya{1}\}$ off-diagonals by evaluating $H(Z|E_1)=1-H_{\text{bin}}(p)$ where $H_{\text{bin}}$ is the binary entropy.  Suppose $W $ is any basis that is complementary (i.e.\ MU) to $Z$, then the distance of $\EC$ to a channel that destroys $W $ off-diagonals is \emph{independent} of $p$, $H(W |E_1)=1$.

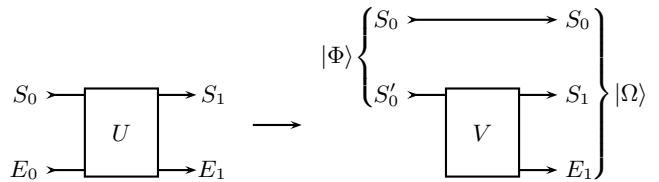
\begin{figure}[t]
\begin{center}
\begin{pspicture}(-4.6,-0.2)(4.2,2.2) 
\newpsobject{showgrid}{psgrid}{subgriddiv=1,griddots=10,gridlabels=6pt}
\def\rectc(#1,#2){%
\psframe[fillcolor=white,fillstyle=solid](-#1,-#2)(#1,#2)}
\def\vertdash(#1){\psline[linestyle=dashed,linewidth=0.01](0.0,0.0)(0.0,#1)}
\def\vtpair{\vertdash(1.0)\rput(0.1,0){\vertdash(1.0)}%
\psline(0.0,0.0)(0.1,0.0)\psline(0.0,1.0)(0.1,1.0)}
\def\vvv#1{\vrule height #1 cm depth #1 cm width 0pt}
\def\rbrac{$\left.\vvv{1.1}\right\}$}
\def\lbrac{$\left\{\vvv{0.59}\right.$}
\psline{>->}(-4.0,1.0)(-2.0,1.0)
\psline{>->}(-4.0,0.0)(-2.0,0.0)
\psline{->}(-1.25,0.6)(-.6,0.6)
\rput(-3.0,0.5){\rectc(0.5,0.6)}
\rput(-3.0,0.5){$U$}
\rput[r](-4.1,1.0){$S_0$}
\rput[r](-4.1,0.0){$E_0$}
\rput[r](-1.6,1.0){$S_1$}
\rput[r](-1.6,0.0){$E_1$}

\psline{>->}(0.8,2.0)(2.8,2.0)
\psline{>->}(0.8,1.0)(2.8,1.0)
\psline{->}(1.8,0.0)(2.8,0.0)
\rput(1.8,0.5){\rectc(0.5,0.6)}
\rput(1.8,0.5){$V$}
\rput(0.26,1.5){\lbrac}
\rput[r](0.7,2.0){$S_0$}
\rput[r](0.7,1.0){$S'_0$}
\rput[r](0.1,1.5){$\ket{\Phi}$}
\rput[l](2.9,2.0){$S_0$}
\rput[l](2.9,1.0){$S_1$}
\rput[l](2.9,0.0){$E_1$}
\rput[l](3.2,1.0){\rbrac}
\rput[l](3.55,1.0){$\ket{\Omega}$}
\end{pspicture}

\caption{%
  How to convert the dynamical evolution of the system into a tripartite pure state.\label{fgr2}}
\end{center}
\end{figure}

With the quadratic measure, we can express the connection between (D1)--(D3) in this paradigm by combining \eqref{eqn13} and \eqref{eqn17} and, e.g., specializing to rank-one projectors $Z_j = \dya{j}$, we find
\begin{equation}
\frac{1}{d_{S_0}^2}\sum_{j,k\neq j} \| \EC(\dyad{j}{k}) \|^2=H_Q(Z|E_1)=\langle \CC_Q (W|S_1)\rangle_{\BC_Z },\notag
\end{equation}
where $d_{S_0}=\dim (\HC_{S_0})$. Thus, the ability of $\EC$ to preserve $Z$ off-diagonals as measured by the Hilbert-Schmidt norm is equal to the ability of the complementary channel $\FC$ to destroy the $Z$ information, measured by $H_Q(Z|E_1)$, which is in turn equal to the ability of $\EC$ to preserve information types $W$ that are complementary to $Z$, measured by $\CC_Q (W|S_1)$ averaged over an equivalence class.

\section{Connection between information-processing tasks}\label{sct5}

Let us now change our focus from decoherence to information theory. The results in Section~\ref{sct3} point to a basic connection between seemingly different information-processing tasks: distinguishing a ``decohered" state from an ``undecohered" state, distilling entanglement from a measurement, and distilling secure (classical) bits. Combining Eqs.~\eqref{eqn9},~\eqref{eqn12}, and \eqref{eqn14}, we find:
$$ D(\rho_{AB}|| \sum_j Z_j \rho_{AB}Z_j)=E^{M_Z|AB}_D= H(Z|C).$$

Let us examine the operational meaning of each of these three quantities. Suppose there exist $n$ physical copies, where $n$ is very large ($n\to\infty$), of the pure state $\tilde{\rho}_{M_Z ABC}=V_Z\rho_{ABC}V_Z\ad$ (defined in Sect.~\ref{sbct2.2}), where Mary, Alice, and Charlie respectively possess the $M_Z$, $AB$, and $C$ portions of each copy. As a first task, suppose that Alice is unaware that her $AB$ systems have been decohered w.r.t.\ $Z$, i.e., are each described by the density operator $\tilde{\rho}_{AB}=\sum_j Z_j \rho_{AB} Z_j$, and her task to determine whether her $AB$ systems are described by $\rho_{AB}$ or by $\tilde{\rho}_{AB}$. To do so, she performs a measurement on the $(AB)^{\ot n}$ system. Assuming she chooses the optimal measurement (see \cite{VedrEtAlPRA97}), then the probability for Alice to confuse the two density operators is \cite{VedrEtAlPRA97}
$$P_n(\sum_j Z_j \rho_{AB}Z_j \to \rho_{AB})=e^{-n D(\rho_{AB}|| \sum_j Z_j \rho_{AB} Z_j)}.$$

As a second (alternative) task, suppose that Alice and Mary wish to distill EPR pairs through local operations and classical communication (LOCC). Then the optimal rate $R^{M_Z |AB}_{\text{EPR}}$ (i.e., EPR pairs per copy) for them to accomplish their task given by $R^{M_Z |AB}_{\text{EPR}}=E^{M_Z|AB}_D$. We note that a one-way hashing protocol \cite{DevWin05} can achieve the optimal rate in this case, since we have $E^{M_Z|AB}_D=-H(M_Z | AB)$ as shown in Appendix~\ref{app5}.

As a third (alternative) task, suppose that Alice measures $Z$ on each of her $A$ copies, and her task is to distill classical bits that are uniformly random as seen by Charlie, by applying universal hashing (i.e. privacy amplification \cite{RenKon05, RennerThesis05}) to her measurement outcomes. We refer to such classical bits as ``secure bits" \cite{RenKon05}, bits that are secure from the adversary Charlie, or bits whose information is completely absent from Charlie's system ($C^{\ot n}$) as defined in Sect.~\ref{sbct2.2}. The optimal rate $R^{Z|C}_{\text{secure}}$ (i.e., secure bits per copy) for Alice to accomplish her task is given by $R^{Z|C}_{\text{secure}}=H(Z|C)$ \cite{RenKon05, Renes2010}. 

Combining the above results for the three different tasks, we find:
\begin{equation}
\label{eqn20}
P_n(\sum_j Z_j \rho_{AB}Z_j \to \rho_{AB}) =e^{-n R^{M_Z |AB}_{\text{EPR}}}=e^{-n R^{Z|C}_{\text{secure}}}.
\end{equation}
Therefore, we have shown that, \textit{in asymptotia}, the tasks of locally determining whether or not the state has been decohered, distilling entanglement from a measurement using LOCC, and distilling secure classical bits using hashing are quantitatively connected.

\section{Discord}\label{sct6}

\subsection{General considerations} \label{sbct6.1}

In this section, we discuss measures of the non-classicality or ``quantumness" of correlations, using the term ``discord" in a general sense to describe any such measure. In particular we discuss how discord measures can be constructed based on each of the four views of decoherence, i.e., based on each of the four classicality conditions appearing in Corollary~\ref{thm4}. 

For most of our discussion below, when referring to the discord of $\rho_{AB}$, we will mean the one-way discord, although we briefly remark on the two-way discord at the end of this section. Henceforth, we will restrict to the case where $Z=\{Z_j\}$ is an orthonormal basis on $\HC_A$, i.e.\ the $Z_j$ are rank-one projectors on system $A$, and for simplicity we will write
$$\EC_Z(\sg_{AB}):= \sum_j Z_j\sg_{AB} Z_j$$
for the $\CC\!\QC$ state obtained from pinching $\sg_{AB}$ in the $Z$ basis.

Though stronger constraints have been considered \cite{BrodModi2011, ModiEtAl2011review}, we impose only that a discord measure should satisfy the two properties that it is:\\

(P1) non-negative, and

(P2) equal to zero iff the state is classically correlated.\\

For one-way discord, ``classically correlated" means a $\CC\!\QC$ state, i.e.\ $\rho_{AB}$ equals $\EC_Z(\rho_{AB})$ for some $Z$. (Later when we mention two-way discord, classically correlated will mean a $\CC\CC$ state.) We consider specific measures satisfying (P1) and (P2) below. 

\subsection{Measures based on (D1) or (D4)}

The literature on discord measures has recently grown at an extraordinary rate; we refer the reader to \cite{ModiEtAl2011review}. While most of the focus has been on measures constructed based on (D1), some recent works \cite{PianiEtAl11, StrKamBru11, GharEtAl2011, PianiAdesso2011} have constructed measures based on (D4). Below we focus primarily on our contribution, which are measures based on (D2) or (D3), but let us mention first a few popular measures based on (D1) or (D4).

The original one-way discord \cite{OllZur01} can be written:
\begin{equation}
\label{eqn21}
\dl^\to(\rho_{AB})=\min_Z [ I(\rho_{AB}) - I(\EC_Z(\rho_{AB}))],
\end{equation}
where $I(\rho_{AB})=H(\rho_A)+H(\rho_B)-H(\rho_{AB})$ is the quantum mutual information. Indeed, $\dl^\to(\rho_{AB})$ has properties (P1) and (P2) \cite{OllZur01}, although (P2) is not obvious. The form of $\dl^\to$ given here is loosely based on (D1), the non-classicality is measured by how far the mutual information is from the mutual information of a $\CC\!\QC$ state.

An obvious way to construct measures based on (D1) is to compute the distance to a $\CC\!\QC$ state. For example, this distance (or more precisely, distinguishability) can be measured by the relative entropy, in which case we obtain another well-studied measure, the one-way information deficit \cite{ZurekDemons03, HorEtAl05}. 
\begin{align}
\label{eqn22}
\Dl^\to(\rho_{AB})&= \min_{\sg_{AB}\in \CC\!\QC} D(\rho_{AB}|| \sg_{AB})\notag\\
&=\min_Z \min_{\sg_{AB}\in \NC} D(\rho_{AB}||  \EC_Z( \sg_{AB}) )\notag\\
&=\min_Z D(\rho_{AB}|| \EC_Z(\rho_{AB}) )\notag\\
&=\min_Z [ H(\EC_Z(\rho_{AB}) )-H(\rho_{AB}) ].
\end{align}
The second line follows by noting that a general state in $\CC\!\QC$ can be written $\EC_Z( \sg_{AB}) $ for some $Z$ and some $\sg_{AB}\in \NC$, the third line follows from \eqref{eqn12}, and the fourth line follows from the discussion in Appendix~\ref{app3}. This measure obviously satisfies (P1) and (P2), and it can be easily shown that $\Dl^\to(\rho_{AB})\geq \dl^\to(\rho_{AB})$ \cite{ZurekDemons03, BroTer10}. Also, by combining \eqref{eqn9}, \eqref{eqn14}, and \eqref{eqn22} we arrive at
\begin{equation}
\label{eqn23}
\Dl^\to(\rho_{AB}) = \min_{Z} E^{M_Z|AB}_D =\min_{Z} E^{M_Z|AB}_R,
\end{equation}
which is a result from \cite{StrKamBru11}, showing that the one-way information deficit can be viewed as a measure based on (D4).

Instead of the relative entropy one can use the Hilbert-Schmidt distance to measure the distance to a $\CC\!\QC$ state, in which case we arrive at the geometric discord from \cite{DakVedBruPRL10}, 
\begin{align}
\label{eqn24}
\Dl_{Q}^\to(\rho_{AB})&= \min_{\sg_{AB}\in \CC\!\QC} D_{\HS}(\rho_{AB}, \sg_{AB}) \notag\\
&= \min_{Z} D_{\HS}(\rho_{AB}, \EC_Z(\rho_{AB}) ).
\end{align}
(We use the subscript $Q$ since it is a quadratic measure.) The geometric discord lower-bounds the one-way information deficit:
\begin{equation}
\label{eqn25}
(\ln 2)\Dl^\to(\rho_{AB})\geq \Dl_Q^\to(\rho_{AB}),
\end{equation}
which follows from $(\ln 2)D(\rho ||\sg)\geq D_{\HS}(\rho ,\sg)$, see Appendix~\ref{app7}. This bound may be useful as quadratic measures are typically easier to calculate than their von Neumann counterparts.

We consider a fourth measure from the literature \cite{StrKamBru11}, based on the geometric entanglement:
\begin{equation}
\label{eqn26}
\Dl_{E_G}^\to(\rho_{AB})= \min_{Z} E^{M_Z|AB}_G.
\end{equation}
While this is obviously based on (D4), it can be rewritten in a form based on (D1) by combining \eqref{eqn15} with \eqref{eqn11}:
\begin{align}
\label{eqn27}
\Dl_{E_G}^\to(\rho_{AB})&= \min_{Z} \min_{\sg_{AB}\in \NC} [1-F(\rho_{AB}, \EC_Z( \sg_{AB}))]\notag\\
&= \min_{\sg_{AB}\in \CC\!\QC}[1- F(\rho_{AB},  \sg_{AB})].
\end{align}
This connection was pointed out in \cite{StrKamBru11}, and we have supplied the proof here.

\subsection{Measures based on (D2)} \label{sbct6.2}

Let us now discuss our main contribution to the understanding of discord. Corollary~\ref{thm4} states that, for pure $\rho_{ABC}$, the $Z$ information being perfectly present in $C$ is a classicality condition for $\rho_{AB}$, i.e.\ $\rho_{AB}$ is a $\CC\!\QC$ state iff there exists a basis $Z$ on $\HC_A$ whose information is known to $C$. Therefore, \textit{the discord of $\rho_{AB}$ must be some quantitative measure for the information about bases on $\HC_A$ that is missing from $C$}. Since conditional entropy measures missing information, we arrive at a new strategy to construct discord measures, by the general form:
\begin{equation}
\label{eqn28}
\DC^\to_K(\rho_{AB}):=\min_Z H_K(Z|C),
\end{equation}
where $H_K(Z|C)$ denotes a conditional entropy with the properties that it is (P1$'$) non-negative and (P2$'$) equal to zero iff $C$ perfectly contains the $Z$ information.\footnote{In order for $\DC^\to_K(\rho_{AB})$ to be well-defined, $H_K(Z|C)$ must also be invariant to local isometries on $C$, since purifications of $\rho_{AB}$ are unique up to isometries on $C$. This property is typical of entropies, so it is not very restrictive; indeed all entropies considered in this article satisfy it.} It is well-known that $H(Z|C)$ satisfies (P1$'$) and (P2$'$), and Appendices~\ref{app1} and~\ref{app2} show that $H_Q(Z|C)$ and $H_{\min}(Z|C)$ also satisfy these properties. (There exist several other entropies that satisfy these properties, e.g.\ the max-entropy \cite{RennerThesis05}, though we do not discuss them here.) For each of these entropies, $H_K(Z|C)=H_K(Z)$ when $\rho_{AB}$ is a pure state, so in this case the discord becomes $\DC^\to_K(\rho_{AB})=\min_Z H_K(Z)$. In the von Neumann case, $\min_Z H(Z)=H(\rho_A)$ is the standard measure of entanglement.

Consider the measure based on von Neumann conditional entropy. Combining \eqref{eqn9} with \eqref{eqn22}, we find that
\begin{equation}
\label{eqn29}
\DC^\to(\rho_{AB}):=\min_Z H(Z|C)= \Dl^\to(\rho_{AB}).
\end{equation}
So we have shown that the one-way information deficit can be viewed as a measure based on (D2). The same can be said for the geometric discord; from \eqref{eqn10} and \eqref{eqn24}, we obtain:
\begin{equation}
\label{eqn30}
\DC^\to_Q(\rho_{AB}):=\min_Z H_Q(Z|C)= \Dl^\to_Q(\rho_{AB}).
\end{equation}

Likewise, the measure based on min-entropy gives:
\begin{align}
\label{eqn31}
\DC^\to_{\min}(\rho_{AB})&:= \min_Z H_{\min}(Z|C)\notag\\
&=-\log_2 [\max_Z p_{\text{guess}}(Z|C)]  \notag\\
&=-\log_2 [\max_{\sg_{AB}\in \CC\!\QC} F(\rho_{AB}, \sg_{AB})]\notag\\
&=-\log_2 [1-\Dl_{E_G}^\to(\rho_{AB})],
\end{align}
showing that $\DC^\to_{\min}$ and $\Dl_{E_G}^\to$ are intimately connected. Equation~\eqref{eqn31} makes it clear that $\DC^\to_{\min}$ and $\Dl_{E_G}^\to$ can be viewed from any of the three perspectives, (D1), (D2), and (D4).

While these connections are interesting, there is a very significant consequence of the fact that discord can be connected to a conditional entropy, and this is because some conditional entropies have operational meanings. Therefore, we are in a position to give a new operational meaning to, e.g., the one-way information deficit.

Thinking of the discord of $\rho_{AB}$ as a resource, one can ask what sort of task would benefit from information missing from the purifying system, which naturally brings to mind cryptography. Here we interpret discord operationally in terms of distillable secure bits \cite{RenKon05}, i.e.\ classical bits that are uniformly random and independent of the purifying system ($C$), where the distillation (a.k.a.\ privacy amplification) is done with universal hashing \cite{RenKon05, RennerThesis05}. (We also briefly discussed secure bits in Sect.~\ref{sct5}.) Consider a scenario where Alice measures the $Z$ basis on $A$ but an adversary, who wants to minimize Alice's secure bits, has control over a local unitary on $A$ just prior to the $Z$ measurement. In the asymptotic case (infinitely many copies of $\rho_{ABC}$) where Alice measures $Z$ on each $A$ copy, the adversary's best strategy is to always choose a unitary on $A$ such that $H(Z|C)$ is as small as possible. It follows \cite{RenKon05, Renes2010} from \eqref{eqn29} that Alice can distill secure bits at an optimal rate of $R^{Z|C}_{\text{secure}}=\Dl^\to(\rho_{AB})$. Now consider a \textit{single shot} version of this scenario.  In this case, Alice performs the $Z$ measurement on just one copy of $\rho_{ABC}$ and the adversary's best strategy is to choose a unitary such that $H_{\min}(Z|C)$ is as small as possible, where we use an interpretation of $H_{\min}$ from \cite{TRSS10}. Then from \cite{TRSS10}, $\DC^\to_{\min}(\rho_{AB})$ approximately quantifies the number of secure bits that Alice can extract through universal hashing. 

We elaborate on the above operational interpretation in Appendix~\ref{app8}. What this discussion implies is that the discord of $\rho_{AB}$ vanishes iff $R^{Z|C}_{\text{secure}}$ vanishes in the above scenario, in other words, iff the minimum distillable secure bits vanishes. In this sense, distillable security is a measure of non-classical correlations and vice-versa.

\subsection{Measures based on (D3)} \label{sbct6.3}

Corollary~\ref{thm4} gives another classicality condition, which may be the most complicated of the four, but nonetheless provides a strategy for constructing discord measures. The classicality condition is that the information about an equivalence class of bases that are MU w.r.t.\ some basis $Z$ (on $\HC_A$) is completely absent from $B$. Our notion of an equivalence class is that two bases (each MU to $Z$) belong to the same class if they can be made equivalent by the action of a unitary that is diagonal in the $Z$ basis. One's intuition is best when $A$ is a qubit, then an equivalence class corresponds to all orthonormal bases that lie in some plane through the origin of the Bloch sphere, e.g.\ the $xy$-plane.

The following is one possible general form for discord measures based on (D3):
\begin{equation}
\label{eqn32}
\DB^\to_K(\rho_{AB}):=\min_Z \langle \CC_K (W|B)\rangle_{\BC_Z},
\end{equation}
again where $\langle \cdot \rangle_{\BC_Z}$ is the average over all bases $W$ in the equivalence class $\BC_Z $. Here $\CC_K (W|B)$ is some measure of the certainty about $W$ given $B$, with the properties that $\CC_K (W|B)\geq 0$, and $\CC_K (W|B)= 0$ iff the $W$ information is completely absent from $B$. Therefore $\DB^\to_K$ is automatically constructed to vanish iff the information about some equivalence class is completely absent from $B$. When $A$ is a qubit, the minimization in \eqref{eqn32} corresponds to minimizing over all planes through the origin of the Bloch sphere. Apparently $B$ having some information about each plane of $A$'s Bloch sphere is a signature of the non-classicality of $\rho_{AB}$.

Using the certainty measures defined in \eqref{eqn8}, we can construct the following discord measures:
\begin{align}
\DB^\to(\rho_{AB})&:= \min_Z \langle \CC (W|B)\rangle_{\BC_Z}\notag \\
\DB^\to_Q(\rho_{AB})&:= \min_Z \langle \CC_Q (W|B)\rangle_{\BC_Z} \notag\\
\DB^\to_{\min}(\rho_{AB})&:= \min_Z \langle \CC_{\min} (W|B)\rangle_{\BC_Z}.\notag 
\end{align}
We note the following connections. First, from \eqref{eqn17} and \eqref{eqn10} it is clear that $\DB^\to_Q $ is just the geometric discord,
$$\DB^\to_Q(\rho_{AB}) = \Dl^\to_Q(\rho_{AB}). $$
So this gives a slightly different perspective on the geometric discord. Also, from the uncertainty relation \eqref{eqn16}, it follows that $\DB^\to$ lower-bounds the one-way information deficit: 
$$\Dl^\to(\rho_{AB})=\min_Z H(Z|C)\geq \DB^\to(\rho_{AB}).$$

The utility of the discord measures constructed from \eqref{eqn32} remains to be determined. We will further explore measures based on (D3) in future work.

\subsection{Two-way discord} \label{sbct6.4}

The two-way discord of $\rho_{AB}$ measures how far the state is from a $\CC\CC$ state, a state whose eigenbasis can be written as a tensor product of bases $Z\ot Z' $ on $A$ and $B$ respectively, i.e.\ of the form $\sum_{j,k} p_{j,k} Z_j \ot Z'_k$. To apply the results in Sect.~\ref{sct3} to this case, one first sets $B$ in these theorems to be a trivial system, then imagines $A$ is a joint system, say $A'B'$, and $Z$ is a type of information about $A'B'$ corresponding to a tensor product of bases. Because all of our main results apply in a similar way, all of the quantitative connections given above for one-way discord have analogs for two-way discord. 

For example, we can introduce a general form for measures of two-way discord based on (D2):
\begin{equation}
\label{eqn33}
\DC^\leftrightarrow_K(\rho_{AB}):=\min_{Z\ot Z'} H_K(Z\ot Z' |C),
\end{equation}
as the minimum information missing from $C$ about a tensor product of bases on $A$ and $B$. We make the connection that $\DC^\leftrightarrow$ corresponds to a popular measure of two-way discord called the relative entropy of quantumness \cite{HorEtAl05},
\begin{align}
\label{eqn34}
\DC^\leftrightarrow&(\rho_{AB}):= \min_{Z\ot Z'} H(Z\ot Z' |C) \notag\\
&= \min_{Z\ot Z'} D(\rho_{AB}|| \sum_{j,k} (Z_j\ot Z'_k) \rho_{AB}(Z_j\ot Z'_k))\notag\\
&= \min_{Z\ot Z'} \min_{\sg_{AB}\in \NC} D(\rho_{AB}|| \sum_{j,k} (Z_j\ot Z'_k) \sg_{AB}(Z_j\ot Z'_k))\notag\\
&= \min_{\sg_{AB}\in\CC\CC} D(\rho_{AB}|| \sg_{AB})\notag\\
&= \min_{Z\ot Z'} E^{M_{ZZ'}|AB}_D =\min_{Z\ot Z'} E^{M_{ZZ'} |AB}_R
\end{align}
The second line follows from \eqref{eqn9}, the third line \cite{ModiEtAl2010} from \eqref{eqn12}, and the fourth line from realizing that the minimization in the third line covers all $\CC\CC$ states. The fifth line is a result from \cite{PianiEtAl11} and can also be seen to follow from \eqref{eqn14}, where $M_{ZZ'}$ is a register that stores the $Z\ot Z'$ information. 
 
Similarly, one can write a (multi-line) equation analogous to \eqref{eqn31} for $\DC^\leftrightarrow_{\min}(\rho_{AB}):= \min_{Z\ot Z'} H_{\min}(Z\ot Z' |C) $, replacing $Z$ with $Z\ot Z'$, $\CC\!\QC$ with $\CC\CC$, and $\Dl_{E_G}^\to(\rho_{AB})$ with $\Dl_{E_G}^\leftrightarrow(\rho_{AB}):=\min_{Z\ot Z'}E^{M_{ZZ'}|AB}_G$.

In a manner analogous to the one-way case, $\DC^\leftrightarrow$ and $\DC^\leftrightarrow_{\min }$ quantify the minimum secure bits distillable from a measurement in a tensor product of bases on $AB$. Indeed this implies an operational interpretation for the relative entropy of quantumness, as the optimal rate to distill secure bits in asymptotia, for the worst-case measurement $Z\ot Z'$ on $AB$.

Finally we note that two-way discord is larger:
\begin{equation}
\label{eqn35}
\DC_K^\leftrightarrow(\rho_{AB})  \geq \DC_K^\to(\rho_{AB}) 
\end{equation}
for each of the three entropies that we have considered. One way to see this is to note that, for each entropy considered, $\DC_K^\leftrightarrow$ can be written as a distance to the set $\CC\CC$, whereas $\DC_K^\to$ is the distance to the set $\CC\!\QC$, and $\CC\CC\subset \CC\!\QC $. A second way to see this is to think of the $Z$ information on $A$ as a coarse-graining of the $Z\ot Z'$ information on $AB$, which is more obvious if one writes the right-hand-side of \eqref{eqn28} as $\min_Z H_K(Z\ot I |C)$. Now we note that each of the entropies considered have the property of decreasing under coarse-grainings, $H_K(X|C)\geq H_K(\hat{X}|C)$ if $\hat{X}$ is a coarse-graining of $X$.  

\section{Conclusion}\label{sct7}

We have presented general quantitative connections (free from any Hamiltonian model) between four phenomena, (D1)--(D4), suggesting of course that they are simply different views of one single phenomenon. The connections, in Sec.~\ref{sct3}, were given in the form of information-theoretic equations, e.g.\  stating that the $Z$-information missing from the environment is a quantitative measure of the entanglement generated in a $Z$-measurement and also the distance of the state to one with no $Z$ off-diagonals. This strong connection between (D1), (D2), and (D4) implied that three information-processing tasks, distilling secure bits, distilling entanglement from a measurement, and distinguishing the state from a decohered one, are intimately related in asymptotia, as discussed in Sect.~\ref{sct5}.

We also connected (D3), the loss of complementary information, to other decoherence definitions. This connection, however, is weaker than the others, primarily because it is formulated using our quadratic measure, which is not operationally motivated. It remains an important open question as to whether the quantitative connection of (D3) to the other decoherence definitions can be established with other measures, perhaps ones that have operational meanings. Towards this end, our approach of considering an equivalence class of bases may be useful.

One of the most powerful aspects of our results in Sec.~\ref{sct3} is that they are valid for every basis of the Hilbert space, allowing one to give a \textit{basis-dependent} description of decoherence from different views. If one chooses to consider the ``most classical" basis, then one arrives at relations for quantum discord. In Sec.~\ref{sct6}, we showed that several popular measures of discord are in fact measuring the information (about this most classical basis) that is missing from the purifying system. This led to a general strategy for making discord measures based on missing information, where one inserts one's favorite measure of missing information (i.e.\ conditional entropy). It is especially interesting that information missing from the purifying system can be stated operationally in terms of the number of secure classical bits that can be distilled, and we showed that several popular measures admit such an operational interpretation. Finally, we showed how discord measures might be constructed based on complementary information.

In conclusion, we have made progress in unifying four different views of decoherence. We hope that other quantitative connections, or perhaps even qualitatively new views of decoherence, are discovered in the future.

\acknowledgments

I thank Michael Zwolak, Li Yu, Vlad Gheorghiu, Shiang Yong Looi, Scott Cohen, Vaibhav Madhok, Roger Colbeck, Wojciech Zurek, Robin Blume-Kohout, and especially Robert Griffiths for helpful discussions. This work was supported by the Office of Naval Research.

\appendix

\section{Properties of $H_Q(Z|C)$}\label{app1}

Let $Z$ be a type of information about system $A$. Denote $N=|Z|$ as the number of elements in $\{Z_j\}$. ($N =d_A$ if $Z$ is an orthonormal basis.) Then the quadratic entropy $H_Q(Z|C)$ has the following properties:\\
(i) $0\leq H_Q(Z|C)\leq  (1-1/N)\Tr(\rho_C^2) $,\\
(ii) $ H_Q(Z|C)=0$ iff the $Z$ information is perfectly present in $C$,\\
(iii) $H_Q(Z|C)= (1-1/N)\Tr(\rho_C^2)$ iff the $Z$ information is completely absent from $C$.

Denote $p_j= \Tr(Z_j\rho_{A})$ and $\sg_{C,j}=\Tr_A(Z_j\rho_{AC})=p_j\rho_{C,j}$. To prove the lower bound in (i), rewrite $H_Q(Z|C)$ as follows,
\begin{align}
\label{eqn36}
H_{Q}(Z|C)&= \Tr(\rho_{C}^2)-\Tr(\tilde{\rho}_{M_ZC}^2) \notag \\
&= \Tr(\rho_{C}^2)-\sum_j \Tr(\sg_{C,j}^2) \notag \\
&= \sum_{j,k} \Tr(\sg_{C,j} \sg_{C,k})-\sum_j \Tr(\sg_{C,j}^2) \notag \\
&= \sum_{j,k\neq j} \Tr(\sg_{C,j} \sg_{C,k}).
\end{align}
The non-negativity follows from $\Tr(\sg_{C,j} \sg_{C,k})\geq 0$. Also, property (ii) follows from the fact that $\Tr(\sg_{C,j} \sg_{C,k})= 0$ if and only if $\sg_{C,j}$ and $\sg_{C,k}$ have orthogonal support.

To prove the upper bound in (i), rewrite $H_Q(Z|C)$ as follows,
\begin{align}
\label{eqn37}
&H_{Q}(Z|C)=\Tr(\rho_{C}^2)-\sum_j \Tr(\sg_{C,j}^2)\notag\\
&= \frac{N-1}{N}\Tr(\rho_C^2)+ \frac{1}{N}\sum_{j,k} \Tr(\sg_{C,j}\sg_{C,k})-\sum_j\Tr(\sg_{C,j}^2) \notag\\
&= \frac{N-1}{N}\Tr(\rho_C^2)- \frac{1}{2N}\sum_{j,k\neq j} D_{\HS}(\sg_{C,j},\sg_{C,k}).
\end{align}
It follows from the last line that $H_Q(Z|C)\leq (1-1/N)\Tr(\rho_C^2)$ because $D_{\HS}(\sg_{C,j},\sg_{C,k})\geq 0$. Also, $D_{\HS}(\sg_{C,j},\sg_{C,k})= 0$ iff $\sg_{C,j}=\sg_{C,k}$, meaning $p_j = \Tr \sg_{C,j}=\Tr \sg_{C,k}=p_k$ and $\rho_{C,j}= \rho_{C,k}$. Therefore, the term $\sum_{j,k\neq j} D_{\HS}(\sg_{C,j},\sg_{C,k})$ is zero iff, for all $j$ and $k$, $p_j=p_k=1/N$ and $\rho_{C,j}= \rho_{C,k}=\rho_C$, i.e. $Z$ is uniformly distributed and independent of $C$. This proves (iii).

Note that combining \eqref{eqn37} with \eqref{eqn8} gives a simple formula for the certainty $\CC_Q$,
\begin{align}
\label{eqn38}\CC_Q(Z|C)&=N \Tr(\tilde{\rho}_{M_ZC}^2)-\Tr(\rho_{C}^2)\\ 
\label{eqn39}&= \sum_{j,k> j} D_{\HS}(\sg_{C,j},\sg_{C,k}).
\end{align}
From \eqref{eqn39}, it is obvious that $\CC_Q(Z|C)=0$ iff the $Z$ information is completely absent from $C$.

\section{Properties of $H_{\min}(Z|C)$} \label{app2}

Since $1/N\leq p_{\text{guess}}(Z|C)\leq 1$, it is clear that $0\leq H_{\min}(Z|C)\leq \log_2 N$. To show that $H_{\min}(Z|C)= \log_2 N$ iff the $Z$ information is absent from $C$, note that $p_{\text{guess}}(Z|C)=1/N$ if the $Z$ information is absent from $C$, and if the $Z$ information is not absent from $C$ then $\log_2 N> H(Z|C) \geq H_{\min}(Z|C)$ \cite{TomColRen09}. To show that $H_{\min}(Z|C)= 0$ iff the $Z$ information is perfectly present in $C$, first note that if all the conditional density operators $\rho_{C,j}$ are orthogonal, then obviously $p_{\text{guess}}(Z|C)=1$. Conversely if $p_{\text{guess}}(Z|C)=1$, then there exists a POVM $\{Q_j\}$ such that $\Tr(Q_j\rho_{C,j})=1$ for each $j$, which implies that the diagonal elements of $Q_j$ are 1 over the support of $\rho_{C,j}$, but since $\sum Q_j = I$ the subspaces over which the diagonal elements of each $Q_j$ are 1 must be orthogonal, which means that the support of each $\rho_{C,j}$ must be orthogonal; this proves the converse.

\section{Proof of Theorem~\ref{thm1}}\label{app3}

In this proof, we consider the pure state $\tilde{\rho}_{M_ZABC}=V_Z \rho_{ABC} V_Z \ad$, noting that $\tilde{\rho}_{AB}=\sum_j Z_j \rho_{AB} Z_j$ is the state that is fully-decohered with respect to $Z$. (In the notation of Sect.~\ref{sct6}, $ \tilde{\rho}_{AB}=\EC_Z(\rho_{AB})$.)

\subsection{Proof of Eq.~\eqref{eqn9}}

The proof of Eq.~\eqref{eqn9} is straightforward:
\begin{align}
H(Z|C)&=H(\tilde{\rho}_{M_ZC})-H(\rho_C)\notag \\
&= H(\tilde{\rho}_{AB})-H(\rho_{AB})\notag\\
& = -\Tr(\rho_{AB}\log_2 \tilde{\rho}_{AB})-H(\rho_{AB})\notag\\
&=  D(\rho_{AB}|| \tilde{\rho}_{AB}).
\label{eqn40}
\end{align}
The third line follows from $\Tr(\sum_j Z_j\rho_{AB} Z_j\log_2 \tilde{\rho}_{AB})= \Tr[\rho_{AB} \sum_j Z_j(\log_2 \tilde{\rho}_{AB}) Z_j]= \Tr(\rho_{AB}\log_2 \tilde{\rho}_{AB})$, since $\log_2 \tilde{\rho}_{AB}$ is an object with no $Z$ off-diagonal elements. Finally, invoke remark \eqref{eqn12}.

\subsection{Proof of Eq.~\eqref{eqn10}}

The proof of Eq.~\eqref{eqn10} is very simple and is similar to the proof of Eq.~\eqref{eqn9}:
\begin{align}
H_{Q}(Z|C)&=\Tr(\rho_{C}^2)-\Tr(\tilde{\rho}_{M_ZC}^2) \notag\\
&= \Tr(\rho_{AB}^2)-\Tr(\tilde{\rho}_{AB}^2) \notag\\
& = \Tr(\rho_{AB}^2)+ \Tr(\tilde{\rho}_{AB}^2)-2\Tr(\tilde{\rho}_{AB} \tilde{\rho}_{AB}) \notag\\
&=  D_{\HS}(\rho_{AB}, \tilde{\rho}_{AB}),
\label{eqn41}
\end{align}
since $\Tr(\sum_j Z_j \rho_{AB}Z_j \tilde{\rho}_{AB})=\Tr(\rho_{AB} \sum_j Z_j \tilde{\rho}_{AB} Z_j)= \Tr(\rho_{AB} \tilde{\rho}_{AB})$. Finally, invoke remark \eqref{eqn12}.

\subsection{Proof of Eq.~\eqref{eqn11}}

We use the fact that the min-entropy is dual to the max-entropy, $H_{\min}(A|C)=-H_{\max}(A|B)$ for any tripartite pure state $\rho_{ABC}$ \cite{KonRenSch09}, where
\begin{align}
&H_{\min}(A|C)= \max_{\sg_C\in \NC}[-\min\{\lm\in \mathbb{R} : \rho_{AC}\leq 2^{\lm} (I\ot \sg_C) \}], \notag \\
&H_{\max}(A|B)=\max_{\sg_B\in \NC}\log_2 F(\rho_{AB}, I \ot \sg_B). \notag
\end{align}
Applying this duality to $\tilde{\rho}_{M_ZABC}$ gives
\begin{align}
\label{eqn42}
&H_{\min}(Z|C):= H_{\min}(M_Z|C) =-H_{\max}(M_Z|AB) \notag \\
&= - \log_2 \max_{\sg_{AB}\in \NC} F(\tilde{\rho}_{M_ZAB} , I \ot \sg_{AB} ) \notag \\
&= - \log_2 \max_{\sg_{AB}\in \NC} F(\tilde{\rho}_{M_ZAB} , V_Z V_Z\ad( I \ot \sg_{AB}) V_Z V_Z\ad ) \notag \\
&= - \log_2 \max_{\sg_{AB}\in \NC} F(\rho_{AB} , V_Z\ad (I \ot \sg_{AB}) V_Z ) \notag \\
&= -\log_2 \max_{\sg_{AB}\in \NC} F(\rho_{AB}, \sum_j Z_j \sg_{AB} Z_j).
\end{align}
The third line uses the fact that $F(\rho,\sg)=F(\rho, \Pi_{\rho}\sg \Pi_{\rho})$ where $\Pi_{\rho}$ is the projector onto the support of $\rho$, and the fourth line uses the fidelity's invariance under isometries. Now, combining this with $H_{\min}(Z|C)=-\log_2 p_{\text{guess}}(Z|C)$ gives the desired result.

\section{Alternative form for $H_Q(Z|C)$} \label{app4}

For any tripartite pure state $\rho_{ABC}$ and any type of information $Z$ about $A$, Eq.~\eqref{eqn13} can be proven as follows:
\begin{align}
&H_Q(Z|C)= D_{\HS}(\rho_{AB} , \sum_j Z_j\rho_{AB}Z_j)\notag\\
&=\Tr(\rho_{AB}^2)-\Tr[(\sum_j Z_j\rho_{AB}Z_j)^2]\notag\\
&=\sum_{j,k,l,m}\Tr[(Z_j\rho_{AB}Z_k)(Z_l\rho_{AB}Z_m)]-\Tr[(\sum_j Z_j\rho_{AB}Z_j)^2]\notag\\
&=\sum_{j,k\neq j,l,m\neq l}\Tr[(Z_j\rho_{AB}Z_k)(Z_l\rho_{AB}Z_m)]\notag\\
&=\sum_{j,k\neq j}\Tr[(Z_j\rho_{AB}Z_k)(Z_k\rho_{AB}Z_j)]\notag\\
&=\sum_{j,k\neq j}\| Z_j\rho_{AB}Z_k \|^2.\notag
\end{align} 
Note that applying this to the first decoherence paradigm (Sect.~\ref{sbst4.1}), where the system $S$ and environment $E$ are described by a pure state $\rho_{SE}$, and specializing to rank-one projectors $Z_j=\dya{j}$, gives a simple form:
$$H_Q(Z|E)=\sum_{j,k \neq j} |\mted{j}{\rho_S}{k}|^2.$$

\section{Proof of Theorem~\ref{thm2}} \label{app5}

The proof of \eqref{eqn14} can be obtained by inspecting the proof of Theorem 1 from \cite{StrKamBru11} and realizing that their proof approach naturally generalizes to the conditions of our theorem. The proof first notes that $\tilde{\rho}_{M_ZABC}$ is a pure state, and so its conditional entropy satisfies $H(Z|C):= H(M_Z|C)  = - H(M_Z | AB)$, but the latter gives a lower bound on the distillable entanglement \cite{DevWin05}, which in turn lower bounds the relative entropy of entanglement
\cite{HHH2000}, so we have:
\begin{align}
\label{eqn43}
H(Z|C)& = - H(M_Z|AB) \notag \\
&\leq E^{M_Z|AB}_D\notag \\
&\leq E^{M_Z|AB}_R\notag \\
&\leq D(\tilde{\rho}_{M_ZAB} || V_Z \tilde{\rho}_{AB} V_Z\ad) \notag \\
&= D(\rho_{AB} || \tilde{\rho}_{AB} ) = H(Z|C),
\end{align}
where the fourth line used the fact that $V_Z \tilde{\rho}_{AB} V_Z\ad$ is separable across the $M_Z|AB$ cut. So the inequalities must be equalities.

Now let us prove Eq.~\eqref{eqn15}. For the geometric entanglement, we have 
\begin{equation}
\label{eqn44}
E^{M_Z|AB}_G =1-\max_{\sg_{M_ZAB}\in \SC}F(\tilde{\rho}_{M_ZAB},\sg_{M_ZAB})
\end{equation}
where $\SC$ is the set of states that are separable across the $M_Z|AB$ cut. Consider a state $\sg_{M_ZAB}=\sum_l p_l \sg_{M_Z,l}\ot \sg_{AB,l} \in \SC$, with $\sg_{AB}=\Tr_{M_Z} (\sg_{M_ZAB})= \sum_l p_l \sg_{AB,l}$. Then
\begin{align}
\label{eqn45}
\sg_{M_ZAB}\leq \sum_l p_l I \ot \sg_{AB,l} = I\ot \sg_{AB}.
\end{align}
Now we use the fact that the fidelity increases upon increases in its second argument, i.e.\ $F(S , T) \leq F(S , T')$ for positive-semidefinite operators $S$ and $T$ with $T' \geq T$. We find that
\begin{align}
\label{eqn46}
&F(\tilde{\rho}_{M_ZAB},\sg_{M_ZAB}) \notag \\
&\leq F(\tilde{\rho}_{M_ZAB}, I\ot \sg_{AB}) \notag \\
&=F(\tilde{\rho}_{M_ZAB}, V_ZV_Z\ad (I\ot \sg_{AB}) V_ZV_Z\ad) \notag \\
& = F(\tilde{\rho}_{M_ZAB} , V_Z \sum_j Z_j\sg_{AB} Z_j V_Z\ad).
\end{align}
In the second step, we noticed that $\tilde{\rho}_{M_ZAB}= V_Z\rho_{AB}V_Z\ad $ lives in the subspace defined by the projector $V_ZV_Z\ad$, and that $F(\rho,\sg)=F(\rho, \Pi_{\rho}\sg \Pi_{\rho})$ where $\Pi_{\rho}$ is the projector onto the support of $\rho$. 

Equation~\eqref{eqn46} implies that for any separable state $\sg_{M_ZAB}$ there exists another separable state $V_Z \sum_j Z_j\sg_{AB} Z_j V_Z\ad$ that is closer (in the fidelity sense) to $\tilde{\rho}_{M_ZAB}$. So we can rewrite \eqref{eqn44},
\begin{align}
\label{eqn47}
E^{M_Z|AB}_G &=1-\max_{\sg_{AB}\in \NC}F(\tilde{\rho}_{M_ZAB}, V_Z \sum_j Z_j\sg_{AB} Z_j V_Z\ad)\notag \\
&=1-\max_{\sg_{AB}\in \NC}F(\rho_{AB}, \sum_j Z_j\sg_{AB} Z_j ) \notag \\
&=1-p_{\text{guess}}(Z|C),
\end{align}
hence proving Eq.~\eqref{eqn15}.

\section{Proof of Theorem~\ref{thm3}} \label{app6}

For a tripartite pure state $\rho_{ABC}$, it is helpful to introduce the following maps (which are somewhat analogous to complementary quantum channels though they do not preserve trace):
\begin{align}
\label{eqn48}
\TC_B(\cdot)&= \Tr_A[(\cdot)\rho_{AB}]\notag\\
\TC_C(\cdot)&= \Tr_A[(\cdot)\rho_{AC}],
\end{align}
which respectively map operators on $A$ to operators on $B$ (operators on $C$). From Appendix B of \cite{Gri05}, it follows that, for any four kets $\ket{j}$, $\ket{k}$, $\ket{l}$, $\ket{m}$ on $\HC_A$,
\begin{equation}
\label{eqn49}
\Tr[\TC_B(\dyad{j}{k}) \TC_B(\dyad{l}{m})]= \Tr[\TC_C(\dyad{j}{m}) \TC_C(\dyad{l}{k})].
\end{equation}
Now let us prove the following lemma.
\begin{lemma}
\label{thm5}
Let $\rho_{ABC}$ be pure, let $\SC_Z$ be the set of all pure states $\ket{\psi}\in \HC_A$ that are unbiased w.r.t.\ the $Z=\{\ket{j}\}$ basis on $\HC_A$. Then:
\begin{equation}
\label{eqn50}
d_A^2\langle \Tr[\TC_B(\dya{\psi})^2] \rangle_{\SC_Z}=\Tr(\rho_B^2)+H_Q(Z|C),
\end{equation}
where $\TC_B$ is defined in \eqref{eqn48} and $\langle\cdot \rangle_{\SC_Z}$ denotes the average over all elements $\ket{\psi}$ in $\SC_Z$.
\end{lemma}
\begin{proof}
If $\ket{\psi}\in\SC_Z$, then in general $\ket{\psi}=\sum_j c_j \ket{j}$ where $c_j=e^{i\phi_j}/\sqrt{d_A}$. Then:
\begin{align}
& d_A^2\langle \Tr[\TC_B(\dya{\psi})^2] \rangle_{\SC_Z} \nonumber \\
&= d_A^2\sum_{j,k,l,m} \langle c_jc_k^*c_lc_m^* \rangle_{\SC_Z} \Tr(\TC_B(\dyad{j}{k}) \TC_B(\dyad{l}{m})) \nonumber \\ 
&= d_A^2\sum_{j,l\neq j} \langle |c_j|^2|c_l|^2 \rangle_{\SC_Z}  \Tr(\TC_B(\dyad{j}{j}) \TC_B(\dyad{l}{l})) \nonumber \\ 
&+ d_A^2 \sum_{j,l\neq j} \langle |c_j|^2|c_l|^2 \rangle_{\SC_Z}  \Tr(\TC_B(\dyad{j}{l}) \TC_B(\dyad{l}{j})) \nonumber \\ 
\label{eqn51}&+ d_A^2 \sum_{j} \langle |c_j|^4 \rangle_{\SC_Z}  \Tr(\TC_B(\dyad{j}{j})^2)\\
&= \sum_{j,l\neq j} \Tr(\TC_B(\dyad{j}{j}) \TC_B(\dyad{l}{l})) \nonumber \\ 
&+ \sum_{j,l\neq j} \Tr(\TC_C(\dyad{j}{j}) \TC_C(\dyad{l}{l})) \nonumber \\ 
\label{eqn52}&+ \sum_{j}  \Tr(\TC_B(\dyad{j}{j})^2)  \\ 
&=\Tr(\TC_B(I)^2) +\Tr(\TC_C(I)^2)-\sum_j \Tr [\TC_C(\dya{j})^2] \nonumber \\ 
\label{eqn53}&= \Tr(\rho_B^2)+\Tr(\rho_C^2)-\sum_j \Tr [\TC_C(\dya{j})^2].
\end{align}
For \eqref{eqn51}, we got rid of all terms for which the product $c_jc_k^*c_lc_m^*$ could take on non-real values, because upon averaging these terms vanish since each complex phase is equally likely. For \eqref{eqn52}, we invoked \eqref{eqn49}. Note that the last two terms in \eqref{eqn53} become $H_Q(Z|C)$.
\end{proof}

In what follows, it is useful to view the above result in an alternative way. Denote by $\UC_Z$ the Abelian group of unitary matrices that are diagonal in the $Z$ basis. These unitaries have the general form: $U=\sum_j e^{i\theta_j} \dya{j}$, where $\{\ket{j}\}$ are the $Z$ basis states, and $e^{i\theta_j}$ are arbitrary phase factors.  Since the general form for a state in $\SC_Z$ is $\ket{\psi}= \sum_j c_j \ket{j}$ with $c_j=e^{i\phi_j}/\sqrt{d_A}$, it is clear that a particular state in $\SC_Z$ can be transformed to any other state in $\SC_Z$ by applying a unitary from $\UC_Z$. So the average above over all states in $\SC_Z$ can alternatively be thought of as choosing a fixed state $ \ket{\psi}$ that is unbiased w.r.t.\ $Z$ and then averaging over all unitaries $U\in\UC_Z$ applied to $ \ket{\psi}$. That is, $\langle \Tr[\TC_B(\dya{\psi})^2] \rangle_{\SC_Z}= \langle \Tr[\TC_B(\dya{\psi})^2] \rangle_{\UC_Z}$ where $\langle \cdot \rangle_{\UC_Z} $ denotes the average over all unitaries $U\in\UC_Z$ applied to $ \ket{\psi}$.

Now let us prove the theorem. Consider a basis $W=\{\ket{w_j}\}$ that is MU w.r.t.\ $Z$. Then from Lemma~\ref{thm5}, we obtain
\begin{align}
H_Q&(Z|C)= d_A^2\langle \Tr[\TC_B(\dya{\psi})^2] \rangle_{\UC_Z}-\Tr(\rho_B^2) \notag\\
\label{eqn54}&= d_A \sum_j \langle \Tr[\TC_B(\dya{w_j})^2] \rangle_{\UC_Z}  - \Tr(\rho_B^2)  \\
\label{eqn55}&= \langle d_A \sum_j \Tr[\TC_B(\dya{w_j})^2]  - \Tr(\rho_B^2) \rangle_{\UC_Z}\\
\label{eqn56}&= \langle \CC_{Q}(W |B) \rangle_{\UC_Z}.
\end{align}
For \eqref{eqn54} we noted that one will obtain the same average regardless of which state $\ket{w_j}$ one chooses to apply the unitary group $\UC_Z$ to, \eqref{eqn55} used that the sum of the averages is the average of the sum, and \eqref{eqn56} used \eqref{eqn38}. 

Now we note that the basis $W'=\{U\ket{w_j}\}$, where $U\in\UC_Z$, is also MU w.r.t.\ $Z$. So the average $\langle \CC_{Q}(W |B) \rangle_{\UC_Z}$ in \eqref{eqn56} is precisely $\langle \CC_{Q}(W |B) \rangle_{\BC_Z}$, the average over the equivalence class of bases that are MU w.r.t.\ $Z$.  

We note that, when defining equivalence classes of mutually unbiased bases, it is common to allow for permutations of the basis indices, so that a basis and its permutation do not belong to different equivalence classes. Our proof can easily accommodate permutations. Simply note that the average in \eqref{eqn54} can be replaced by $\langle \cdot \rangle_{\UC_Z, \PC}$, where now the average is over $\UC_Z$ and the set $\PC$ of all permutations of the $j$ indices.

\section{Relation between $D(\rho ||\sg)$ and $D_{\HS}(\rho ||\sg)$} \label{app7}

It is well-known \cite{SchWes02} that $(\ln 2)D(\rho || \sg)\geq 2[D_{T}(\rho,\sg)]^2 $ (the so-called Pinsker inequality), where $D_{T}(\rho,\sg)=(1/2)\Tr |\rho-\sg |$ is the trace distance. So we just need to prove that $2[D_{T}(\rho,\sg)]^2 \geq D_{\HS}(\rho,\sg)$. Denote $\rho-\sg=F-G$ where $F$ and $G$ are orthogonal ($FG=0$) positive operators with $\Tr F=\Tr G \leq 1$. Note $\Tr F^2 \leq (\Tr F)^2$ and likewise for $G$, so $D_{\HS}(\rho,\sg)= \Tr F^2 + \Tr G^2 \leq (\Tr F)^2+(\Tr G)^2= 2 (\Tr F)^2=2[D_{T}(\rho,\sg)]^2$.

\section{Discord as distillable secure bits} \label{app8}

Our interpretation of discord as distillable secure bits is based on (1) the fact the von Neumann entropy $H(Z|C)$ measures the optimal rate for distilling secure bits in the asymptotic i.i.d.\ (identical, independently distributed) case \cite{RenKon05, Renes2010} and (2) the fact that the min-entropy $H_{\min}(Z|C)$ approximately measures the distillable secure bits in the single-shot case \cite{TRSS10}. We note that $Z$ is an orthonormal basis in this Appendix.

Consider first the asymptotic case (infinitely many copies of $\rho_{ABC}$), where Alice measures $Z$ on each $A$ copy but an adversary, who wants to minimize Alice's distillable secure bits, controls a local unitary on $A$ just prior to the $Z$ measurement. While it might seem obvious that the adversary should always choose the unitary that minimizes $H(Z|C)$, let us investigate this. Suppose that, instead of always choosing the same unitary, the adversary chooses $U_i$ with probability $p_i$, which effectively corresponds to Alice measuring in the $X_i=U_i\ad ZU_i$ basis. Suppose they do this for $n$ rounds, then for large $n$ Alice will measure $X_i$ approximately $p_i n$ times, and her overall measurement on $A^{\ot n}$ is of the form $\widetilde{X}=X_1^{\ot p_1n}\ot X_2^{\ot p_2n} \ot ...$.  Now we can group together blocks of $n$ rounds and say that Alice measures $\widetilde{X}$ on each block, and apply the result that $H(\widetilde{X} |C^{\ot n})/n$ measures the optimal rate for distilling secure bits in the asymptotic limit. Because of the tensor-product nature of the measurement $\widetilde{X}$ and of the state $\rho_{ABC}^{\ot n}$ the optimal rate is additive: $H(\widetilde{X} |C^{\ot n})/n=\sum_i p_i H(X_i|C)$. Thus, it follows from this analysis that the adversary should indeed always choose the same unitary, namely the unitary that minimizes $H(Z|C)$. (Choosing a different unitary with non-zero probability would increase Alice's optimal rate.) So we assume the adversary chooses the basis that minimizes $H(Z|C)$ each round, from which it follows that $\Dl^\to(\rho_{AB})=\min_Z H(Z|C)$ quantifies the optimal rate for Alice to distill secure bits through privacy amplification.

Now consider the single-shot case where Alice measures $Z$ on $A$ for a single copy of $\rho_{ABC}$, and the adversary chooses a local unitary on $A$ just prior to the $Z$ measurement to minimize Alice's distillable secure bits. Here we use a result from \cite{TRSS10}. They consider a hash function, pulled at random from a two-universal family, that maps $Z$ to a bit-string of length $\ell $, and they define $\Dl$ (not to be confused with our discord) as a measure of how far the output bit-string is from being perfectly secure from $C$. (See \cite{TRSS10} for the precise definition of $\Dl$ in terms of the trace distance.) They find that, on average, 
\begin{equation}
\label{eqn57}
\Dl = \frac{1}{2}\sqrt{2^{\ell-H_{\min}(Z|C)}}.
\end{equation}
In our scenario, this implies that the adversary's best strategy is to choose a unitary such that Alice measures in a basis with the smallest value of $H_{\min}(Z|C)$, which would maximize the distance of Alice's bit-string from being perfectly secure. Another way to say this is to suppose that Alice wishes for her bit-string to have a particular value of $\Dl$, then choosing the basis with the smallest value of $H_{\min}(Z|C)$ will force Alice to obtain a bit-string with the smallest length $\ell $. We assume the adversary uses the best strategy, and since $\DC^\to_{\min}(\rho_{AB})=\min_Z H_{\min}(Z|C)$, we find: 
\begin{equation}
\label{eqn58}
\DC^\to_{\min}(\rho_{AB}) = \ell +2\log_2 (1/2\Dl).
\end{equation}
This says that the discord $\DC^\to_{\min}(\rho_{AB})$ and the length of the secure bit-string $\ell$ grow in proportion to each other, and are approximately equal up to the term $2\log_2 (1/2\Dl)$.

The above results for one-way discord have an analogous formulation for two-way discord. Consider again the asymptotic case, but now Alice measures $Z\ot Z'$ on each copy, i.e.\ $Z$ on each $A$ and $Z'$ on each $B$, and the adversary now controls a local unitary on $A$ and a local unitary on $B$ just prior to Alice's measurements. Then $\DC^{\leftrightarrow}(\rho_{AB})=\min_{Z\ot Z'} H(Z\ot Z'|C)$ quantifies the optimal rate for Alice to distill secure bits through privacy amplification. Now consider the single shot case where Alice measures $Z\ot Z'$ on a single copy of $\rho_{AB}$, and the adversary controls a local unitary on $A$ and a local unitary on $B$ just prior to Alice's measurement. Like the one-way result above, one finds that $\DC^{\leftrightarrow}_{\min}(\rho_{AB}) = \ell +2\log_2 (1/2\Dl)$, where $\ell$ is the length of the bit-string that is a distance $\Dl$ from perfectly secure that Alice can extract through universal hashing applied to her $Z\ot Z'$ measurement outcomes.

\bibliography{Dec}

\end{document}